\documentclass[10pt, conference, compsocconf]{IEEEtran}
\IEEEoverridecommandlockouts

\usepackage{graphicx}
\usepackage{balance}  
\usepackage{lcsect}

\usepackage{array}

\usepackage{amsmath,amsthm,amssymb}

\makeatletter
\newif\if@restonecol
\makeatother

\usepackage[ruled,vlined,linesnumbered]{algorithm2e}
\usepackage{qtree}
\usepackage{subfigure}

\usepackage{multirow}
\usepackage{hhline}
\usepackage{booktabs}

\usepackage[usenames,dvipsnames]{pstricks}
\usepackage{epsfig}
\usepackage{pst-grad} 
\usepackage{pst-plot} 

\newtheorem{theorem}{Theorem}
\newtheorem{lemma}{Lemma}
\newtheorem{definition}{Definition}
\newtheorem{example}{Example}
\newtheorem{fact}{Fact}

\newcommand{\qedsymb}{\hfill{\rule{2mm}{2mm}}}

\def\0{\phantom{0}}

\newcommand{\remove}[1]{}

\newbox\tallstrutbox
\setbox\tallstrutbox=\hbox{\vrule height10pt depth 3.5pt width0pt}
\def\tallstrut{\relax\ifmmode\copy\tallstutbox\else\unhcopy\tallstrutbox\fi}

\newbox\tallerstrutbox
\setbox\tallerstrutbox=\hbox{\vrule height14pt depth 3.5pt width0pt}
\def\tallerstrut{\relax\ifmmode\copy\tallerstutbox\else\unhcopy\tallerstrutbox\fi}

\usepackage{etex}

\usepackage{listings}
\usepackage{tikz}
\newcommand*\circled[1]{\tikz[baseline=(char.base)]{
  \node[shape=circle,draw,inner sep=0.7pt] (char) {#1};}}

\DeclareMathOperator{\sus}{\mathit{SUS}}

\DeclareMathOperator{\sa}{\mathit{SA}}
\DeclareMathOperator{\rank}{\mathit{Rank}}
\DeclareMathOperator{\lcp}{\mathit{LCP}}

\DeclareMathOperator{\lr}{\mathit{LR}}
\DeclareMathOperator{\llr}{\mathit{LLR}}
\DeclareMathOperator{\llrc}{\mathit{LLRc}}

\begin{document}

\title{On Stabbing Queries for Generalized Longest
  Repeat\thanks{A preliminary version of this work appeared as 
a regular paper in the Proceedings of 
IEEE International Conference on Bioinformatics and Biomedicine
(BIBM), November 9--12, 2015, Washington D.C., USA.}}

\author{Bojian Xu\\
Department of Computer Science, Eastern Washington University,
WA99004, USA.\\
\texttt{bojianxu@ewu.edu}
}

\maketitle

\begin{abstract}
  A longest repeat query on a string, motivated by its applications in
  many subfields including computational biology, asks for the longest
  repetitive substring(s) covering a particular string position (point
  query).  In this paper, we extend the longest repeat query from
  point query to \emph{interval query}, allowing the search for
  longest repeat(s) covering any position interval, and thus
  significantly improve the usability of the solution. Our method for
  interval query takes a different approach using the insight from a
  recent work on \emph{shortest unique
    substrings}~\cite{HPT-spire2014}, as the prior work's approach for
  point query becomes infeasible in the setting of interval query.
  Using the critical insight from~\cite{HPT-spire2014}, we propose an
  indexing structure, which can be constructed in the optimal $O(n)$
  time and space for a string of size $n$, such that any future
  interval query can be answered in $O(1)$ time. Further, our solution
  can find \emph{all} longest repeats covering any given interval
  using optimal $O(occ)$ time, where $occ$ is the number of longest
  repeats covering that given interval, whereas the prior $O(n)$-time
  and space work can find only one candidate for each point query.
  Experiments with real-world biological data show that
  our proposal is competitive with prior works,
  both time and space wise, 
  while providing with the 
  new functionality of interval queries as opposed to 
  point queries provided by prior works.

\end{abstract}


\section{Introduction}
\label{sec:intro}
Repetitive structures and regularity finding in genomes and proteins
is important as these structures play important roles in the
biological functions of genomes and proteins~\cite{Gus97}.  One of the
well-known features of DNA is its repetitive structure, especially in
the genomes of eukaryotes. Examples are that overall about one-third
of the whole human genome consists of repeated
substrings~\cite{McC1993}; about 10--25\% of all known proteins have
some form of repetitive structures~\cite{LW06}. In addition, a number
of significant problems in molecular string analysis can be reduced
to repeat finding~\cite{Mar83}. Therefore, it is of great interest for
biologists to find such repeats in order to understand their
biological functions and solve other problems.

There has been an extensive body of work on repeat finding in the
communities of bioinformatics and stringology.  The notion of maximal
repeat and super maximal
repeat~\cite{Gus97,BDH09,KVX-tcbb2012,BBO-spire2012} captures all the
repeats of the whole string in a space-efficient manner. Maximal
repeat finding over multiple strings and its duality with minimum
unique substrings were also
understood~\cite{BIMSTTT2007,ISY2009,IS2011}. We refer readers
to~\cite{Gus97} (Section~7.11) for the discussion and further pointers
to other types of repetitive structures, such as palindrome and tandem
repeat. However, all these notions of repeats do not track the
locality of each repeat, and thus it is difficult for them to support
position-specific queries (stabbing queries) in an efficient manner.

Because of this reason, longest repeat query was recently proposed and
asks for the longest repetitive substring(s) that covers a
particular string
position~\cite{IKX-repeat-CORR2015,SOG2012-IC,TX-GPU-DASFAA2015}.
Because any substring of a repetitive substring is also
repetitive, longest repeat query effectively provides a ``stabbing''
tool for finding most of the repeats that cover any particular
string position.
The algorithm by Schnattinger et al.~\cite{SOG2012-IC} for computing
bidirectional matching statistics can be used to compute the
\emph{rightmost} longest repeat covering every string position,
whereas the study by \.{I}leri et al.~\cite{IKX-repeat-CORR2015} can
find the \emph{leftmost} longest repeat for every string position.
Both solutions use  optimal $O(n)$ time and space for
finding the longest repeat for all the $n$ string positions.  By
storing the pre-computed longest repeats of every position, they are
able to answer any future longest repeat query in $O(1)$ time, and
thus achieve the amortized $O(1)$ time cost in finding the longest
repeat of any arbitrary string position.
Since it is not clear how to parallelize the optimal algorithms
in~\cite{IKX-repeat-CORR2015,SOG2012-IC}, the recent
study in~\cite{TX-GPU-DASFAA2015} proposed a time sub-optimal but
parallelizable algorithm, so as to take advantage of the modern
multi-processor computing platforms such as the general-purpose
graphics processing units.
\remove{
Because there are only $n$ distinct string positions, by saving the
pre-computed results using $O(n)$ space, all three
works~\cite{IKX-repeat-CORR2015,SOG2012-IC,TX-GPU-DASFAA2015} are
able to answer any future query for the longest repeat covering any
particular string position in $O(1)$ time.
}

\section{Problem Statement}
\label{sec:prob}

We consider a {\bf string} $S[1.. n]$,
 where each character $S[i]$ is
drawn from an alphabet $\Sigma=\{1,2,\ldots, \sigma\}$.
A {\bf substring} $S[i.. j]$
of $S$ represents $S[i]S[i+1]\ldots S[j]$ if $1\leq i\leq j \leq n$,
and is an empty string if $i>j$.
String $S[i'.. j']$ is a {\bf proper substring} of another string
$S[i.. j]$ if $i\leq i' \leq j' \leq j$ and $j'-i' < j-i$. 
%
%
The {\bf length} of a non-empty substring $S[i.. j]$, denoted as
$|S[i.. j]|$, is $j-i+1$. We define the length of an empty string
as zero. 
A {\bf prefix} of $S$ is a substring $S[1.. i]$
for some $i$, $1\leq i\leq n$. 
A {\bf proper prefix} $S[1.. i]$ is a prefix of $S$ where $i <
n$.
A {\bf suffix} of $S$ is a substring
$S[i.. n]$ for some $i$, $1\leq i\leq n$.  
A {\bf proper suffix} $S[i.. n]$ is a suffix of $S$ where $i >
1$.
We say character $S[i]$ occupies the string {\bf position} $i$.
We say substring $S[i.. j]$ {\bf covers} the position interval
$[x.. y]$ of
$S$, if $i\leq x \leq y  \leq j$.  
In the case $x=y$,  we say 
 substring $S[i.. j]$ {\bf covers} the  position $x$ (or $y$) of string
$S$.  
For two strings $A$ and $B$, we write ${\bf A=B}$ (and say $A$ is {\bf
  equal} to $B$), if $|A|= |B|$ and $A[i]=B[i]$ for 
$i=1,2,\ldots, |A|$.  
%
We say $A$ is lexicographically smaller than $B$,
denoted as ${\bf A < B}$, if (1) $A$ is a proper prefix of $B$, or (2)
$A[1] < B[1]$, or (3) there exists an integer $k > 1$ such that
$A[i]=B[i]$ for all $1\leq i \leq k-1$ but $A[k] < B[k]$.
A substring
$S[i.. j]$ of $S$ is {\bf unique}, if there does not exist
another substring $S[i'.. j']$ of $S$, such that 
$S[i.. j] = S[i'.. j']$ but $i\neq i'$. 
A character $S[i]$ is a {\bf singleton}, if it is unique.
A substring is a {\bf repeat} if it is not unique.

\begin{definition}
\label{def:lr}
A {\bf longest repeat (LR)} covering string position interval 
  $[x.. y]$, denoted
as $\lr_x^y$, is 
a repeat substring $S[i.. j]$, such that: (1) $i\leq x\leq y \leq j$, and 
(2) there does not exist  another repeat substring $S[i'.. j']$, such
that $i'\leq x \leq y \leq j'$ and $j'-i' > j-i$. 
\end{definition}

Obviously, for any string position interval $[x.. y]$, if
$S[x.. y]$ is not unique, $\lr_x^y$ must exist, because at least
$S[x.. y]$ itself is a repeat.  Further, there might be multiple
\emph{choices} for $\lr_x^y$. For example, if $S={\tt abcabcddbca}$, then
$\lr_2^3$ can be either $S[1.. 3]={\tt abc}$ or $S[2.. 4]={\tt
  bca}$.

\noindent{\bf Problem} (generalized stabbing LR query).
  Given a string position interval $[x.. y]$, $1\leq x \leq y \leq
  n$, find all choices of $\lr_x^y$ or the fact that it
  does not exist.

We call the generalized stabbing LR query as \emph{interval query},
which includes the \emph{point query} as a special case where
$x=y$. All prior works~\cite{SOG2012-IC,IKX-repeat-CORR2015,TX-GPU-DASFAA2015} only
studied point query.  {\em Our goal is to find an efficient mechanism for
finding the longest repeats of every possible string position
interval.}

\section{Prior Work and Our contribution}
\label{sec:prior}
In addition to the related work discussed in Section~\ref{sec:intro},
there were recently a sequence of work on finding \emph{shortest
  unique substrings}
(SUS)~\cite{PWY-ICDE2013,TIBT2014,IKX-CPM2014,ikx-sus-tcs2015,HPT-spire2014},
of which Hu et al.~\cite{HPT-spire2014} studied the generalized
version of SUS finding: 
\emph{Given a string position interval $[x.. y]$, $1\leq x \leq y \leq
  n$, find $\sus_x^y$, the shortest unique substring that
  covers the string position interval $[x..y]$, or the fact that such $\sus_x^y$
  does not exist.
}

To the best of our knowledge, no efficient reduction from LR
finding to SUS finding is known as of now. That is, given a set of SUSes covering a set of
position intervals respectively, it is not
clear how to find the set of LRs that cover that same set of
position intervals respectively, by only using 
the string $S$, the given set of SUSes, and 
linear (of the set
size) time cost for the reduction. The reason 
behind the hardness of obtaining such an efficient reduction is because simply
chopping off one ending character of an SUS does not necessarily produce an LR. 

For example: suppose $S={\tt a..aba..a}$ of $2n+1$ characters,
where every character is $\texttt{a}$ except the middle one is ${\tt
  b}$. Clearly, $\sus_{n-1}^n=S[n-1,n+1] = {\tt aab}$, whereas
$\lr_{n-1}^n=S[1..n]$. Given $\sus_{n-1}^n$ and $S$ itself, it is
not clear how to find $\lr_{n-1}^n=S[1..n]$ using $O(1)$ time, without involving
other auxiliary data structures (otherwise, the reduction, which is
still unknown, can become so
complex, making itself no better than a self-contained solution for finding
LR, which is what this 
paper is presenting.).

Due to the overall importance of repeat finding in bioinformatics and
the lack of efficient reduction from SUS finding to LR finding, it is
our belief that providing and implementing a complete solution for
generalized LR finding will be beneficial to the community. In
summary, we make the following contributions.

\noindent
1. We generalize the longest repeat query from point query to
  \emph{interval query}, allowing the search for the longest repeat(s)
  covering any interval of string positions, and thus
  significantly improve the usability of the solution.

\noindent
2. Because there are at most $n$ point queries for a string of size
$n$, all prior works pre-compute and save the results of every
possible point query, such that any future point query can be answered
in $O(1)$ time.  However, in the setting of interval queries, there
are $\binom{n}{2}+n = \Theta(n^2)$ distinct intervals.  It becomes
impossible, under the $O(n)$ time and space budget, to achieve the
amortized $O(1)$ query response time, by pre-computing and storing the
longest repeats covering each of the $\Theta(n^2)$ intervals.
Therefore, a different approach is needed. Our approach uses the insight
from the work by HU et al.~\cite{HPT-spire2014} that leads us to 
an indexing structure, which can
be constructed using optimal $O(n)$ time and space, such that, by
using this indexing structure, any future interval query can still be
answered in $O(1)$ time. The $O(n)$ time and space costs are optimal
because reading and saving the input string already needs $O(n)$ time
and space.

\noindent
3. Our work can find all longest repeats covering any given
  interval using  optimal $O(occ)$ time, where $occ$ is the number
  of the longest repeats covering that interval.  However, the work
  in~\cite{IKX-repeat-CORR2015} and~\cite{SOG2012-IC} can only find
  the leftmost and the rightmost candidate, respectively, and only
  support point queries.  The algorithm in~\cite{TX-GPU-DASFAA2015} can
  find all longest repeats covering a string position, but
  their parallelizable sequential algorithm is sub-optimal in the time
  cost ($O(n^2)$, indeed) and only supports point queries as well.

\noindent
4. We provide a generic implementation of our solution without
assuming the alphabet size, making the software useful for the
analysis of different
types of strings.  Experimental study with real-world biological
data shows that our proposal is competitive with prior works, both
time and space wise, while supporting interval
queries in the meantime.

\section{Preparation}
\label{sec:prep}
The {\bf suffix array} $\sa[1.. n]$ of the string $S$ is a
permutation of $\{1,2,\ldots, n\}$, such that for any $i$ and $j$,
$1\leq i < j \leq n$, we have $S[\sa[i].. n] < S[\sa[j].. n]$.
That is, $\sa[i]$ is the start position of the $i$th suffix in
the sorted order of all the suffixes of $S$.
The {\bf rank array} $\rank[1.. n]$ is the inverse of the suffix
array. That is, $\rank[i]=j$ iff $\sa[j]=i$. 
The {\bf longest common prefix (lcp) array} $\lcp[1.. n+1]$ is an
array of $n+1$ integers, such that for $i=2,3,\ldots, n$, $\lcp[i]$ is
the length of the lcp of the two suffixes $S[\sa[i-1].. n]$ and
$S[\sa[i].. n]$. We set $\lcp[1]=\lcp[n+1]=0$.\footnote{In literature,
the lcp array is often defined as an array of $n$ integers. We include
an extra zero at $\lcp[n+1]$ as a sentinel  to simplify the description 
of our upcoming
algorithms.}  
The following table
shows the suffix array and the lcp
array of an example string $S={\tt mississippi}$.
%

\begin{center}
\def\0{\phantom{0}}
{\footnotesize
\begin{tabular}{c|c|c|l}
\hline 
\tallstrut$i$ & $\lcp[i]$  & $\mathit{\sa}[i]$ & suffixes\\
\hline
\hline
\tallstrut$\01$ & $0$ & $11\0$  &{\tt i}\\
$\02$ & $1$ & $\08\0$  & {\tt  ippi}\\
$\03$ & $1$ & $\05\0$  & {\tt  issippi}\\
$\04$ & $4$ & $\02\0$  & {\tt  ississippi}\\
$\05$ & $0$ & $\01\0$  & {\tt  mississippi}\\
$\06$ & $0$ & $10\0$  & {\tt  pi}\\
$\07$ & $1$ &  $\09\0$  & {\tt ppi}\\
$\08$ & $0$ & $\07\0$  & {\tt sippi}\\
$\09$ & $2$  & $\04\0$  & {\tt sissippi}\\
$10$ & $1$  & $\06\0$  & {\tt ssippi}\\
$11$ & $3$ & $\03\0$  & {\tt ssissippi}\\
$12$ & $0$ & -- & --\\
\hline
\end{tabular}
}
\end{center}

\begin{definition}
\label{def:llr}
The
{\bf left-bounded longest repeat (LLR)} starting at position $k$,
denoted as $\llr_k$, is a repeat $S[k.. j]$,
such that either $j=n$ or $S[k.. j+1]$ is unique. 
\end{definition}

Clearly, for any string position $k$, if $S[k]$ is not a singleton,
$\llr_k$ must exist, because at least $S[k]$ itself is a repeat.
Further, if $\llr_k$ does exist, it must have only one choice, because
$k$ is a fixed string position and the length of $\llr_k$ must be as
long as possible.

Lemma~\ref{lem:llr} shows that, by using the rank array and
the lcp array of the string $S$, it is easy to calculate any $\llr_i$ if
it exists or to detect the fact that it does not exist.

\begin{lemma}[\cite{IKX-repeat-CORR2015}]
\label{lem:llr}
For $i=1,2,\ldots,n$: 
$$
\llr_i = 
\left \{
\begin{array}{ll}
S[i.. i + L_i-1], & \textrm{if $L_i > 0$}\\
\textit{does not exist}, & \textrm{if $L_i = 0$}
\end{array}
\right.
$$
where $L_i = \max\{\lcp[\rank[i]],\lcp[\rank[i]+1]\}$.
\end{lemma}

\remove{
\begin{proof}
  Note that $L_i$ is the length of the lcp between the suffix
  $S[i.. n]$ and any other suffix of $S$.  If $L_i > 0$, 
  $S[i.. L_i-1]$ is the lcp among $S[i.. n]$ and any
  other suffix of $S$, so $S[i.. L_i-1]$ is $\llr_i$.  Otherwise
  ($L_i = 0$), the letter $S[i]$ is a singleton, so $\llr_i$ does not
  exist.
\end{proof}
}

Observe that an LLR can be a substring (proper suffix, indeed) of
another LLR. For example, suppose $S={\tt ababab}$, then
$\llr_4=S[4.. 6]={\tt bab}$, which is a substring of
$\llr_3=S[3.. 6] = {\tt abab}$.  Formally, the neighboring
LLRs have the following relationship.

\begin{lemma}[\cite{TX-GPU-DASFAA2015}]
\label{lem:llr-length}
$|\llr_i|\leq |\llr_{i+1}|+1$
\end{lemma}

\remove{
\begin{proof}
 The claim is clearly correct for the case where 
  $|\llr_i|$ is $0$ or $1$, so we only consider the case where
  $|\llr_i| \geq 2$. Suppose $\llr_i = S[i.. j]$, $i < j$. It
  follows that $i+1 \leq j$. Since $S[i.. j]$ is a repeat, its
  substring $S[i+1.. j]$ is also a repeat. Note that $\llr_{i+1}$
  is the longest repetitive substring starting from  position $i+1$, so
  $|\llr_{i+1}| \geq |S[i+1.. j]| = |\llr_i|-1$, i.e., 
$|\llr_i|\leq |\llr_{i+1}|+1$.    
\end{proof}
}

\begin{definition}
\label{def:useful}
We say an LLR is \emph{useless} if it is
a substring of another LLR; otherwise, it is \emph{useful}.
\end{definition}


\begin{lemma}
\label{lem:lr-llr}
Any existing longest repeat $\lr_x^y$, $1\leq x
\leq y \leq n$, must be a useful LLR.
\end{lemma}

\begin{proof}
(1) We first prove $\lr_x^y$ must be an LLR.
Assume that $\lr_x^y=S[i.. j]$ is not an LLR. Note that $S[i..
j]$ is a repeat starting from position $i$. If $S[i.. j]$ is not
an LLR, it means $S[i.. j]$ can be extended to some position
$j' > j$, so that $S[i.. j']$ is still a repeat and  also covers
the position interval $[x.. y]$. That says, $|S[i.. j']| > |S[i.. j]|$.
However, the contradiction is that $S[i.. j]$ is already the longest repeat
covering the position interval $[x.. y]$. (2) Further, $\lr_x^y$ must
be a useful LLR, because if it is a useless LLR, it means there exists
another LLR that covers the position interval $[x.. y]$ but is
longer than
$\lr_x^y$, which contradicts the fact that $\lr_x^y$ is
the longest repeat that covers the interval $[x.. y]$.
\end{proof}

\section{LR finding for one interval }
In this section, we propose an algorithm that takes as input a string position
interval and returns the LR(s) covering that interval. The algorithm
spends $O(n)$ time and space per query but does not need any indexing
data structure. We present this algorithm here in case the
practitioners have only a small number of interval queries of their
interest and thus this light-weighted algorithm will suffice.  We
start with the finding of the leftmost LR covering the given
interval and will give a trivial extension in the end for finding all LRs
covering the given interval.

\begin{lemma}
\label{lem:llr-cover}
For any $i$, $j$, $x$, and $y$, $1\leq i < j \leq x \leq y \leq n$:
If $\llr_j$ does not exist
or exists but does not cover the interval $[x.. y]$,
$\llr_i$  does not exist or does not cover $[x.. y]$
\end{lemma}

\begin{proof}
  We prove the lemma by contradiction.  (1) Assume it is possible that
  when $\llr_j$ does not cover the interval $[x.. y]$, $\llr_i$
  can still cover $[x.. y]$.  Say, $\llr_i = S[i.. k]$ for
  some $k\geq y$. It follows that $S[j.. k]$ is also a repeat and
  covers $[x.. y]$, which is a contradiction, because 
  $\llr_j$, the longest repeat starting from string location $j$, does
  not cover $[x.. y]$.
  (2) Assume it is possible that when $\llr_j$ does not exist,
  $\llr_i$ can still cover $[x.. y]$.  Say, $\llr_i = S[i..
  k]$ for some $k\geq y$. It follows that $S[j.. k]$ is also a
  repeat and covers $[x.. y]$, which is a contradiction, because
$\llr_j$ does not exist at all, i.e., $S[j]$ is a
  singleton.    
\end{proof}

By Lemma~\ref{lem:lr-llr}, we know any LR must be an LLR, so we can
find $\lr_x^y$ covering a given interval $[x.. y]$ by simply
checking each $\llr_i$, $i\leq x$, and picking the longest one
that covers the interval $[x.. y]$. Ties are resolved
by picking the leftmost choice.  Because of Lemma~\ref{lem:llr-cover},
early stop is possible to make the procedure faster in practice by
checking every $\llr_i$ in the decreasing order of the value of $i=x, x-1, \ldots,
1$: the search will  stop whenever we see an  
$\llr_i$ that does not cover the interval $[x.. y]$ or does not exist at all. 
 Algorithm~\ref{algo:one-leftmost} shows the pseudocode, which
returns $(start,end)$, representing the start and ending
positions of $\lr_x^y$, respectively.  If $\lr_x^y$ 
does not exist, $(-1,-1) $ is returned.

\begin{algorithm}[t]
{\scriptsize
  \caption{\footnotesize Find the leftmost $\lr_x^y$ covering a given string position
    interval \textrm{$[x.. y]$}.}
\label{algo:one-leftmost}
\KwIn{(1) Two integers $x$ and $y$, $1\leq x \leq y \leq n$, representing 
a string position interval $[x.. y]$. \newline (2) The rank array and 
      the lcp array of the string $S$.} 
\KwOut{The leftmost $\lr_x^y$ or the fact that $\lr_x^y$ does not exist.}

\smallskip 

$start \leftarrow -1$; 
$end \leftarrow -1$; 
\tcp*{start and end positions of $\lr_x^y$}

\For{$i = x$ down to $1$\label{line:for}}{
  $L \leftarrow \max\{\lcp[\rank[i]],\lcp[\rank[i]+1]\}$\tcp*{$|\llr_i|$}

  \lIf
     {$L=0$ or $i+L-1<y$}{break\tcp*{Early stop}\label{line:break}}
  \ElseIf(\tcp*[f]{Pick the leftmost one})
{$L\geq end-start+1$\label{line:tie}}
{$start\leftarrow i$; $end\leftarrow i+L-1$\label{line:tie2}
}
}
\Return{$\lr_x^y\leftarrow (start,length)$\;}
}
\end{algorithm}

\begin{lemma}
\label{lem:one-leftmost}
Given the rank array and the lcp array of the string $S$, for any
string position interval $[x.. y]$,
Algorithm~\ref{algo:one-leftmost} can find $\lr_x^y$ or the fact that
it does not exist, using $O(x)$ time and $O(n)$ space.  If there are
multiple choices for $\lr_x^y$, the leftmost one is returned.
 \end{lemma}

 \begin{proof}
   The algorithm clearly has no more than $x$ iterations and each iteration
   takes $O(1)$ time, so it costs  $O(x)$ time. The space
   cost is primarily from the rank array and the lcp array, which
   altogether is $O(n)$, assuming each integer in these arrays costs a
   constant number of memory words.  
  If multiple LRs cover position interval $[x.. y]$, the leftmost
  LR will be returned, as is guaranteed by Line~\ref{line:tie} of
  Algorithm~\ref{algo:one-leftmost}. 
 \end{proof}

\begin{theorem}
\label{thm:one-leftmost}
For any position interval $[x.. y]$ in the string $S$, we
can find $\lr_x^y$ or the fact that it does not exist using $O(n)$ time
and space.  If there are multiple choices for $\lr_x^y$, the leftmost
one is returned.
\end{theorem}

\begin{proof}
  The suffix array of $S$ can be constructed by existing algorithms
  using $O(n)$ time and space (e.g., \cite{KA-SA2005}). After the
  suffix array is constructed, the rank array can be trivially created
  using another $O(n)$ time and space.  We can then use the suffix array and
  the rank array to construct the lcp array using another $O(n)$ time
  and space~\cite{KLAAP01}.  Given the rank array and the lcp array,
  the time cost of Algorithm~\ref{algo:one-leftmost} is $O(x)$
  (Lemma~\ref{lem:one-leftmost}). So altogether, we can find $\lr_x^y$ or the
  fact that it does not exists using $O(n)$ time and space.
If there are multiple choices for $\lr_x^y$, the leftmost
choice will be returned, as is claimed in Lemma~\ref{lem:one-leftmost}. 
\end{proof}


\begin{algorithm}[h]
{\scriptsize
  \caption{Find all LRs that cover a given string position interval
    \textrm{$[x .. y]$}}
\label{algo:one-all}
\KwIn{(1) Two integers $x$ and $y$, $1\leq x \leq y \leq n$, representing 
a string position interval $[x .. y]$. \newline (2) The rank array and 
      the lcp array of the string $S$.} 
\KwOut{All LRs that cover the position interval \textrm{$[x .. y]$}
  or the fact that no such LR exists.}

\smallskip 

\tcc{Find the length of $\lr_x^y$.}
$length \leftarrow 0$\;
\For{$i = x$ down to $1$}{
  $L \leftarrow \max\{\lcp[\rank[i]],\lcp[\rank[i]+1]\}$\tcp*{$|\llr_i|$}

  \If
 {$L=0$ or $i+L-1<y$}{break\tcc*{$\llr_i$ does not exist or does not
     cover $[x.. y]$, so we can early stop.}}
  \lElseIf{$L > length$}
   {$length\leftarrow L$\;}
}


\smallskip 

\tcc{Find all LRs that cover position interval $[x.. y]$.}

\If(\tcp*[f]{$\lr_x^y$ does exist.}){$length>0$}{
  \For{$i = x$ down to $1$}{
    $L \leftarrow \max\{\lcp[\rank[i]],\lcp[\rank[i]+1]\}$\tcp*{$|\llr_i|$}
   \If
     {$L=0$ or $i+L-1<y$}{break\tcp*{Early stop}}
    \ElseIf{$L = length$}
    {Print $\lr_x^y\leftarrow (i,i+length-1 )$\;}
  }
}
\lElse{Print $\lr_x^y\leftarrow (-1,-1)$\tcp*{$\lr_x^y$ does not exist.}}
}
\end{algorithm}

\noindent{\bf Extension: find all LRs covering a given position interval.} 
It is trivial to extend Algorithm~\ref{algo:one-leftmost} to find all
the LRs covering any given position interval $[x.. y]$ as
follows. We can first use a similar procedure 
as Algorithm~\ref{algo:one-leftmost} to calculate
$|\lr_x^y|$. If $\lr_x^y$ does exist, then we will start over
the procedure again to re-check every $\llr_i$, $i\leq x$, and return
every LLR whose length is equal to $|\lr_x^y|$.  Due to
Lemma~\ref{lem:llr-cover}, the same early stop as what we have in
Algorithm~\ref{algo:one-leftmost} can be used for practical
speedup. Algorithm~\ref{algo:one-all} shows the pseudocode of this
procedure, which
clearly spends an extra $O(x)$ time. Using
Theorem~\ref{thm:one-leftmost}, we have:

\begin{theorem}
\label{thm:one-all}
For any position interval $[x.. y]$ in the string $S$, we
can find all choices of $\lr_x^y$ or the fact that $\lr_x^y$ does not exist, using
 $O(n)$ time and space.
\end{theorem}

\section{A geometric perspective of the useful LLRs and the LR queries}
\label{sec:geo}

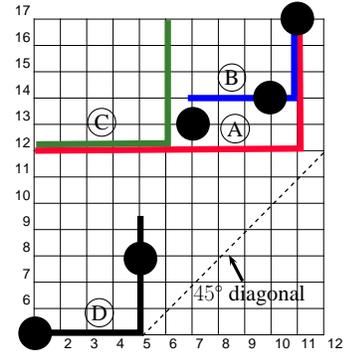
\begin{figure}[t]
 \centering
\scalebox{0.35} 
{
\begin{pspicture}(0,-6.64)(13.24,6.66)
\definecolor{color0c}{rgb}{0.5019607843137255,0.5019607843137255,0.5019607843137255}
\definecolor{color1132}{rgb}{1.0,0.0,0.2}
\rput(0.12,-11.02){\psgrid[gridwidth=0.0182,subgridwidth=0.014111111,gridlabels=15.0pt,subgriddiv=1,subgridcolor=color0c](1,5)(1,5)(12,17)}
\psline[linewidth=0.05cm,linestyle=dashed,dash=0.16cm 0.16cm](5.24,-6.06)(12.12,0.9)
\psline[linewidth=0.24](1.0,-5.94)(5.16,-5.94)(5.16,-1.5)
\psline[linewidth=0.24,linecolor=color1132](1.16,0.98)(11.24,1.06)(11.2,6.1)
\psline[linewidth=0.24,linecolor=blue](6.96,2.98)(11.0,2.98)(11.0,6.06)
\psline[linewidth=0.24,linecolor=OliveGreen](1.2,1.22)(6.2,1.26)(6.2,5.94)
\usefont{T1}{ptm}{m}{n}
\rput(3.597871,-5.18){\Huge \pscirclebox[linewidth=0.04]{D}}
\usefont{T1}{ptm}{m}{n}
\rput(8.632285,3.76){\Huge \pscirclebox[linewidth=0.04]{B}}
\usefont{T1}{ptm}{m}{n}
\rput(3.6856446,2.06){\Huge \pscirclebox[linewidth=0.04]{C}}
\usefont{T1}{ptm}{m}{n}
\rput(8.754726,1.82){\Huge \pscirclebox[linewidth=0.04]{A}}
\usefont{T1}{ptm}{m}{n}
\rput(9.300019,-4.52){\Huge $45^\circ$ diagonal}
\psline[linewidth=0.12cm,arrowsize=0.05291667cm 2.0,arrowlength=1.4,arrowinset=0.4]{->}(9.04,-3.98)(8.52,-2.98)
\psdots[dotsize=1.28](1.16,-5.96)
\psdots[dotsize=1.28](11.12,6.0)
\psdots[dotsize=1.28](10.08,3.0)
\psdots[dotsize=1.28](7.16,2.0)
\psdots[dotsize=1.28](5.16,-3.12)
\end{pspicture} 
}
\caption{The 2d geometric perspective on the useful LLRs of string
  $S={\tt aaababaabaaabaaab}$ and its several generalized LR queries.  (A)
  The $\llrc$ array saves all the useful LLRs in the strictly
  increasing order of their string positions: $\{(1,5), (5,8), (7,13),
  (10,14), (11,17)\}$, where each useful LLR is a $(start,end)$ tuple,
  representing the start and ending position of the LLR.  By viewing
  the $start$ and $end$ positions as the $x$ and $y$ coordinates, all
  the useful LLRs of the example string can be visualized as the dark
  dots in the figure.  (B) Queries for $\lr_{11}^{12}$,
  $\lr_{11}^{14}$, $\lr_{6}^{12}$ and $\lr_{5}^{5}$ are visualized by
  the red, blue, green, and black polylines, numbered
  \protect\circled{A}--\protect\circled{D}, respectively.  (C) All
  dark dots and polylines are on or above the $45^{\circ}$ diagonal.}
\label{fig:geo}
\end{figure}

In this section, we present a geometric perspective of the useful
LLRs and the generalized LR queries.
This perspective is sparked by the idea presented
in~\cite{HPT-spire2014},
which
serves as the
intuition behind the algorithms
in Sections~\ref{sec:2d} and~\ref{sec:1d} 
 that share the similar spirit of those for SUS finding in~\cite{HPT-spire2014}.
We start with the following lemma that says the useful
LLRs are easy to compute.

\begin{lemma}
\label{lem:comp-useful}
Given the lcp and rank arrays of the string $S$, 
we can compute its useful LLRs in $O(n)$ time and space. 
\end{lemma}

\begin{proof}
  By Lemma~\ref{lem:llr-length}, we know if $\llr_{i-1}$ exists,
  the right boundary of $\llr_i$ is on or after the right boundary of
  $\llr_{i-1}$, for any $i\geq 2$, so we can construct the array of
  useful LLRs in one pass as follows: we  calculate each
  $\llr_i$ using Lemma~\ref{lem:llr}, for $i=1,2,\ldots,n$, and 
  eliminate (useless) $\llr_i$, if $|\llr_i| = 0$ or $|\llr_i| =
  |\llr_{i-1}|-1$. 
\end{proof}

\begin{definition} {\bf $\llrc$} is an array of useful LLRs, which
  are saved in the ascending order of their start position. We use
  $\llrc.size$ to denote the number of elements in  $\llrc$.
\end{definition}

Algorithm~\ref{algo:useful} shows the procedure for the $\llrc$ array
construction in $O(n)$ time and space, provided with the suffix array
and lcp array of  $S$. Each $\llrc$ array element is a
$(start,end)$ tuple, representing the start and ending positions of
the useful LLR. Because no useful LLR is a substring of another
useful LLR, we have the following fact.

 \begin{algorithm}[t]
{\scriptsize
  \caption{\footnotesize The calculation of LLRc, the array of useful
    LLRs, saved in ascending order of their positions.}
\label{algo:useful}
\KwIn{The rank and lcp arrays of the string $S$.} 

\smallskip 

$j\leftarrow 1$; 
$prev \leftarrow 1$\;

\For{$i=1\ldots n$}{
  $L \leftarrow \max\{\lcp[\rank[i]],\lcp[\rank[i]+1]\}$\tcp*{$|\llr_i|$}

\smallskip

  \If(\tcp*[f]{$\llr_i$ is useful.}){$L > 0$ and $L \geq prev$}{
    $\llrc[j] \leftarrow (i, i+L-1)$;
    $j\leftarrow j+1$  
  }

\smallskip 

  $prev \leftarrow L$\;
}
\Return{$\llrc$\;}
}
 \end{algorithm}

\begin{fact}
\label{fact:useful}
All elements in the $\llrc$ array have their both start and ending
positions in strictly increasing order. 
That is,  
for any $i$ and
$j$, $1\leq i < j \leq \llrc.size$: 
$\llrc[i].start <
\llrc[j].start$ and $\llrc[i].end < \llrc[j].end$.
\end{fact}

If we view each useful LLR's start position as the $x$ coordinate and
ending position as the $y$ coordinate, each useful LLR can be viewed
as a dot in the 2d space. All the 2d dots, representing all the useful
LLRs that are saved in the LLRc array, are distributed in the 2d space
from the low-left corner toward the up-right corner. Because of
Fact~\ref{fact:useful},
no two dots share the same
$x$ or $y$ coordinates. Further, since every
dot's $y$ coordinate is no less than its $x$ coordinate, those dots
are on or above the $45^\circ$ diagonal. Figure~\ref{fig:geo} shows this geometric
perspective of several useful LLRs.

\begin{definition}
\label{def:weight}
The \emph{weight} of a dot $(x,y)$, 
representing a useful $\llr_x=S[x..y]$, is $|\llr_x| = y-x+1$, the length of  $\llr_x$.
\end{definition}

\begin{definition}
\label{def:subspace}
$S_{x,y} = \{(a,b)\in \llrc \mid a \leq x, b \geq y\}$.
\end{definition}

If we draw in the 2d space
a {\LARGE
  $\lrcorner$} shaped orthogonal polyline whose angle locates at position $(x,y)$,
$S_{x,y}$ is the set of 2d dots, representing those 
useful LLRs that are located on the up-left side (inclusive) of 
the polyline.

Because any LR must be useful LLR (Lemma~\ref{lem:lr-llr}), from 
this geometric perspective, the answer
to the $\lr_x^y$ query becomes the heaviest dot(s), whose horizontal
coordinate is $\leq x$ and whose vertical coordinate is
$\geq y$. That is, $\lr_x^y$ are the heaviest dots
in  $S_{x,y}$. If $S_{x,y}$ is empty, it means $\lr_x^y$ does
not exist. Figure~\ref{fig:geo} shows this geometric perspective of
several generalized LR queries.

\section{AN INDEX OF {\large $O(occ \cdot \log n)$} QUERY TIME}
\label{sec:2d}

\begin{algorithm}[t]
{\scriptsize
  \caption{\footnotesize  Find LR using 2d DMQ.}
\label{algo:2d}
\KwIn{The lcp and rank arrays of the string $S$}

\smallskip 

Compute the LLRc array\tcp*{Algorithm~\ref{algo:useful}}

Build the 2d DMQ index for the LLRc array elements 
\tcp*{Existing technique, e.g.,~\cite{ST-2011pods}}

\smallskip

\tcc{Find one choice of $\lr_x^y$.}
\underline{QueryOne2d($x,y$)}: 

{2dDMQ($x,y$)\tcp*{return $(-1,1)$, if $S_{x,y} = \emptyset$.}}

\smallskip 

\tcc{Find all choices of $\lr_x^y$.}
\underline{QueryAll2d($x,y$)}:

 $(x',y')\leftarrow$ 2dDMQ($x,y$)\;

 \If{$(x',y')\neq (-1,-1)$}{FindAll2d($x,y,y'-x'+1$) \tcp*{Recursive
     searches start.}}

\smallskip

\underline{FindAll2d($x,y,weight$)}: \tcp*[f]{Helper function}

 $(x',y')\leftarrow$ 2dDMQ($x,y$)\;

\If{$(x',y')=(-1,-1)$ or $(y'-x'+1 < weight)$}{\Return\tcp*{Recursion exits.}}

Print $(x',y')$ \tcp*{One choice of $\lr_x^y$ is found.}

\If{$x'-1\geq 1$}{
  FindAll2d($x'-1,y,weight$) \tcp*{New recursive search.}
}

\If{$y'+1\leq n$}{
  FindAll2d($x,y'+1,weight$) \tcp*{New recursive search.}
}

}
\end{algorithm}

As is explained in Section~\ref{sec:geo}, $\lr_x^y$ is the heaviest
dot(s) from the set $S_{x,y}$, if $S_{x,y}$ is not empty; otherwise,
$\lr_x^y$ does not exist. Finding \emph{one} heaviest dot from $S_{x,y}$ is
nothing but the well-known 2d dominance max query.

\noindent{\bf 2d dominance max query (DMQ).} 
Given a set of $n$ dots and any position $(x,y)$ in the 2d space, find
the heaviest dot, whose horizontal coordinate is $\leq x$ and vertical
coordinate is $\geq y$. If there are multiple choices,
ties are resolved arbitrarily.

There exist indexing structures (e.g.,~\cite{ST-2011pods}) that
can be constructed on top of the $n$ dots using $O(n\log n)$ time and
$O(n)$ space, such that by using the indexing structure, any future 2d
DMQ can be answered in
$O(\log n)$ time. 
The reduction from finding an LR  to a 2d DMQ 
immediately gives us the {\tt QueryOne2d} function in
Algorithm~\ref{algo:2d} for finding one choice of an LR.

\remove{
\paragraph{Example (Figure~\ref{fig:geo}):} 
Query $\circled{A}$ is to search for $\lr_{11}^{12}$. That is to
find one of the heaviest dots in $S_{11,12}$, which can be either dot
$(7,13)$ or dot $(11,17)$. So, $\lr_{11}^{12}$ can be either $S[7\dots
13]$ or $S[11.. 17]$, depending on which one is returned by the 2d
DMQ.
Query $\circled{D}$ is to search for
$\lr_5^5$. That is to find one of the heaviest dots in $S_{5,5}$,
which is the dot $(1,5)$. So, $\lr_5^5$ is $S[1\dots 5]$.
}

\begin{theorem}
\label{thm:2d-one}
  We can construct an indexing structure for a string $S$ of size $n$
  using $O(n\log n)$ time and $O(n)$ space, such that by using the
  indexing structure any future generalized $\lr$
  query can be answered in $O(\log n)$ time. If there exist multiple
  choices for the $\lr$ of interest, ties are resolved
  arbitrarily.
\end{theorem}

\begin{proof}
  (1)  The  suffix  array  of  $S$  can  be  constructed  by  existing
  algorithms     using     $O(n)$     time    and     space     (e.g.,
  \cite{KA-SA2005}). After  the suffix array is  constructed, the rank
  array can be trivially created  using $O(n)$ time and space.  We can
  then use  the suffix array and  the rank array to  construct the lcp
  array using another $O(n)$ time and space~\cite{KLAAP01}.  (2) Given
  the rank array and the lcp array, we can construct the $\llrc$ array
  of     useful    LLRs     using    $O(n)$     time     and    space
  (Lemma~\ref{lem:comp-useful}  and Algorithm~\ref{algo:useful}).  (3)
  We then  create the indexing  structure for the LLRc  array elements
  for 2d DMQ, using $O(n\log n)$ time and $O(n)$ space
  (e.g.,  \cite{ST-2011pods}). By using  this index, we  can answer any
  future generalized  LR  query in  $O(\log  n)$ time  and  ties  are  resolved
  arbitrarily.  
\end{proof}

\subsection{Find all choices of any LR.}
We know $\lr_x^y$ are the heaviest dots in $S_{x,y}$, if $S_{x,y}$ is
not empty; otherwise, $\lr_x^y$ does not exist.  Upon receiving a
query for $\lr_x^y$, we first perform a 2d DMQ,
which returns one of the heaviest dots in $S_{x,y}$. If no
such a dot is returned, then $\lr_x^y$ does not exist. Otherwise, suppose $(x',
y')$ is the dot returned, then $(x',y')$ is one of the choices for
$\lr_x^y$. 

Because all the dots representing the LLRc array elements have their
both $x$ and $y$ coordinates strictly
increase~(Fact~\ref{fact:useful}), all other choices (if existing) of
$\lr_x^y$ must be existing in the union of $S_{x'-1,y}$ and $S_{x,
  y'+1}$. Therefore, we can find other choices of $\lr_x^y$ by the
following two recursive searches: one will find one of the
heaviest dots in $S_{x'-1,y}$, the other will find one of the heaviest
dots in $S_{x, y'+1}$. Each of these two recursive searches is again a
2d DMQ.

For each recursive search: (1) If the weight of
the heaviest dot it finds is equal to $y'-x'+1$, the length of
$\lr_x^y$, it will return the found dot as another choice of $\lr_x^y$ and
will then launch its own two new recursive searches, similar to what
its caller has done in order to find  other choices for
$\lr_x^y$; (2) otherwise, it stops and returns to its caller.

Function  {\tt QueryAll2d} in Algorithm~\ref{algo:2d} shows the pseudocode
for finding all choices of $\lr_x^y$.

\begin{example}[Figure~\ref{fig:geo}]
Search $\circled{A}$ is for $\lr_{11}^{12}$. That is to find
all heaviest dots in $S_{11,12}$, which include dot $(7,13)$ and dot
$(11,17)$. Suppose the 2d DMQ launched by search
\circled{A} returns dot $(7,13)$, which has a weight of $7$ and 
is {\bf one choice for $\lr_{11}^{12}$}. The next two recursive
searches launched by search \circled{A} will be search \circled{B}
looking for one of the heaviest dots in $S_{11,14}$ and search
\circled{C} looking for one of the heaviest dots in $S_{6,12}$.

Search \circled{B} will return the heaviest dot $(11,17)$ from
$S_{11,14}$, whose weight is equal to $7$, so the dot $(11,17)$ is
{\bf another choice of
  $\lr_{11}^{12}$}. Search \circled{B} will then launch its own two new
recursive searches for one heaviest dot in each of $S_{10,14}$ and
$S_{11,18}$. (These two searches are not shown in Figure~\ref{fig:geo}
for concision). The search in $S_{10,14}$ returns dot $(10,14)$ whose
weight is less than $7$, so the search stops and returns to its
caller. The search in $S_{11,18}$ finds nothing, so it stops and
returns to its caller. After all its recursive searches return, search
\circled{B} returns to its caller, which is search \circled{A}.

Search \circled{C} finds nothing in $S_{6,12}$, so it stops
and returns to its caller, which is search \circled{A}.

At this point, all the work of search $\circled{A}$ is finished,
and we have found all the choices, which are $S[7.. 13]$ and
$S[11.. 17]$ (or $\llrc[3]$ and $\llrc[5]$, equivalently), for $\lr_{11}^{12}$.
\end{example}

Clearly, the same 2d DMQ index is used in finding all
choices of an LR query, and there are no more than $2 \cdot occ+1$
instances of 2d DMQ, in the finding of all choices
of an LR, where $occ$ is the number of choices of the LR.  Because
each 2d DMQ takes $O(\log n)$ time, we get the
following theorem.

\begin{theorem}
  We can construct an indexing structure for a string $S$ of size $n$
  using $O(n\log n)$ time and $O(n)$ space, such that by using the
  indexing structure, we can find all choices of any LR in $O(occ\cdot
  \log n)$ time, where $occ$ is the number of choices of the LR being queried for.
\end{theorem}

\section{AN INDEX OF {\large $O(occ)$} QUERY TIME}
\label{sec:1d}
In this section, we present the optimal indexing structure for generalized LR
finding. It is again based on the intuition derived from the geometric
perspective on the relationship between useful LLRs and LR queries
(Section~\ref{sec:geo}).

Recall that the answer for an $\lr_x^y$ query is the heaviest dot(s)
from $S_{x,y}$, if $S_{x,y}$ is not empty. Due to
Fact~\ref{fact:useful}, $S_{x,y}$ corresponds to a continuous chunk of the 
LLRc array, if $S_{x,y}$ is not empty. Therefore, searching for one
heaviest dot in $S_{x,y}$ becomes searching for one heaviest
element within a continuous chunk of the LLRc array, which is nothing
but the \emph{range minimum query} on the array LLRc.\footnote{We 
should actually perform \emph{range maximum query}, which however can be
  trivially reduced to RMQ by viewing each array
  element as the negative of its actual value.}

\noindent{\bf Range minimum query (RMQ).} Given an array $A[1..
n]$ of $n$ comparable elements, find the index of the smallest element
within $A[i.. j]$, for any given $i$ and $j$, $1\leq i \leq j \leq
n$. If there are multiple choices, ties are resolved arbitrarily.

There exist indexing structures (e.g.,~\cite{FH-2006CPM,HT-84SICOMP})
that can be constructed on top of the array $A$ using $O(n)$ time and
space, such that any future RMQ can be answered in $O(1)$ time. 

The next issue is: Upon receiving a query for $\lr_x^y$, for some
$x$ and $y$, $1\leq x \leq y \leq n$, how to find the left and right
boundaries of the continuous chunk of LLRc, over which we will
perform an RMQ ? Due to
Fact~\ref{fact:useful} and with the aid of the geometric perspective
of the useful LLRs, we can observe that the left boundary of the
chunk only depends on the value of $y$, whereas the right boundary of the chunk only
depends on the value of $x$. Intuitively, if one sweeps a horizontal
line starting from position $y$ (inclusive) toward the up direction,
the LLRc array index of the first dot that the line hits is the left
boundary of the RMQ's range. Similarly, if one sweeps a vertical line
starting from position $x$ (inclusive) toward the left direction, the
LLRc array index of the first dot that the line hits is the right boundary
of the RMQ's range. 
The range for RMQ is invalid, if any one of the following three
possibilities happens: 1) No dot is hit by the horizontal line; 2) No
dot is hit by the vertical line; 3) The index of the left boundary of
the range is larger than the index of the right boundary of the range.
An invalid RMQ range means that $\lr_x^y$ does not exist.
See Figure~\ref{fig:geo} for examples. 

More precisely, given the values of
$x$ and $y$ from the query for $\lr_x^y$, the left boundary $L_y$ and the right boundary $R_x$ of
the range for RMQ can be determined as follows:
\remove{
$$
L_y = \left \{
\begin{array}{ll}
\min\{i \mid \llrc[i].end \geq y\},\\ \ \ \ \ \ \ \ \ \ \ \ \ \ \textrm{\ if\ }  \{i \mid \llrc[i].end \geq y\} \neq \emptyset\\
-1, \ \ \ \ \ \ \ \ \textrm{\ otherwise}
\end{array}
\right .\ \ \ \ \ \ \ \ \ \ \ \ \ 
$$

$$
R_x = \left \{
\begin{array}{ll}
\max\{i \mid \llrc[i].start \leq x\}, \\
\ \ \ \ \ \ \ \ \ \ \ \ \  \textrm{\ if\ }  \{i \mid \llrc[i].start \leq x\} \neq \emptyset\\
-1,  \ \ \ \ \ \ \ \  \textrm{\ otherwise}
\end{array}
\right .
$$
}
\remove{
$$
L_y = \left \{
\begin{array}{l}
\min\{i \mid \llrc[i].end \geq y\},\\ 
  \hspace{21mm}\textrm{\ if\ }  \{i \mid \llrc[i].end \geq y\} \neq
  \emptyset\\
-1, \hspace{15mm} \textrm{\ otherwise}
\end{array}
\right .
$$

$$
R_x = \left \{
\begin{array}{l}
\max\{i \mid \llrc[i].start \leq x\}, \\
 \hspace{21mm} \textrm{\ if\ }  \{i \mid \llrc[i].start \leq x\} \neq \emptyset\\
-1, \hspace{15mm}  \textrm{\ otherwise}
\end{array}
\right .
$$
}
{\small
$$
L_y \hspace*{-1mm}= \hspace*{-1mm}\left \{
\begin{array}{ll}
\hspace*{-2mm}\min\{i \mid \llrc[i].end \geq y\}, &
  \hspace*{-1mm}\textrm{\ if\ }  \{i \mid \llrc[i].end \geq y\} \neq
  \emptyset\\
\hspace*{-2mm}-1, &  \hspace*{-1mm}\textrm{\ otherwise}
\end{array}
\right .
$$
$$
R_x \hspace*{-1mm}= \hspace*{-1mm} \left \{
\begin{array}{ll}
\hspace*{-2mm}\max\{i \mid \llrc[i].start \leq x\}, &
 \hspace*{-4mm} \textrm{\ if\ }  \{i \mid \llrc[i].start \leq x\} \neq \emptyset\\
\hspace*{-2mm}-1, & \hspace*{-4mm}  \textrm{\ otherwise}
\end{array}
\right .
$$
}
Further, we can pre-compute $L_y$ and $R_x$, for every 
$x=1,2,\ldots,n$ and $y=1,2,\ldots,n$, and save the results for future
references. Algorithm~\ref{algo:range} shows the procedure for 
computing the $L$ and $R$ arrays, which clearly uses $O(n)$ time
and space.

\begin{lemma} 
\label{lem:range} 
Algorithm~\ref{algo:range} computes $L_1, L_2, \ldots, L_n$ and\\ $R_1,
R_2, \ldots, R_n$ using $O(n)$ time and space. 
\end{lemma}

\begin{algorithm}[t]
{\scriptsize
  \caption{\footnotesize Compute $L_i$ and $R_i$ for $i=1,2,\ldots, n$.}
\label{algo:range}
\KwIn{The $\llrc$ array.}
\KwOut{The $L$ and $R$ arrays.}

\smallskip 

\lFor{$i=1\ldots n$}{$L_i \leftarrow -1$; $R_i \leftarrow -1$\tcp*{Initialization.}}

$i \leftarrow 1$\; 

\For{$y=1\ldots n$}{
  \lIf{$y \leq \llrc[i].end$}{
    $L_y \leftarrow i$\;
  }
  \lElseIf{$i < \llrc.size$}{
   $i\leftarrow i+1$;  $L_y \leftarrow i$\;  
  }
  \lElse{break\;}
}

$i \leftarrow \llrc.size$\; 

\For{$x=n \ldots 1$}{
  \lIf{$x \geq \llrc[i].start$}{
    $R_x \leftarrow i$\;
  }
  \lElseIf{$i > 1$}{
   $i\leftarrow i-1$;  $R_x \leftarrow i$\;  
  }
  \lElse{break\;}
}

}
\end{algorithm}

Now we are ready to present the algorithm for finding one choice of a
generalized LR query. Algorithm~\ref{algo:rmq} (through
Line~\ref{line:one}) gives the pseudocode. After array LLRc is
created, we will compute the $L$ and $R$ arrays using the LLRc array
(Algorithm~\ref{algo:range}). Then we will create the RMQ structure
for the LLRc array, where the weight of each array element is defined
as the length of the corresponding LLR (or, from the geometric
perspective, is the weight of the 2d dot representing that LLR), using
existing techniques (e.g.,~\cite{FH-2006CPM,HT-84SICOMP}). Upon
receiving a query for $\lr_x^y$, function {\tt QueryOneRMQ(x,y)}
performs an RMQ over the range $\llrc[L_y,R_x]$, if $1\leq L_y \leq
R_x\leq n$; otherwise, it returns $(-1,-1)$, meaning $\lr_x^y$ does
not exist. The answer returned by the RMQ is  one of the choices for
$\lr_x^y$. If there exist multiple choices for $\lr_x^y$, ties are
resolved arbitrarily, depending on which heaviest element in the range
is returned by the RMQ.

\begin{algorithm}[t]
{\scriptsize
  \caption{\footnotesize  Find LR using RMQ.}
\label{algo:rmq}
\KwIn{The lcp and rank arrays of the string $S$.}

\smallskip 
Compute the LLRc array\tcp*{Algorithm~\ref{algo:useful}}
Compute the $L$ and $R$ arrays from the LLRc array \tcp*{Algo.~\ref{algo:range}}
Construct the RMQ structure for the LLRc array\tcp*{\cite{FH-2006CPM,HT-84SICOMP}}

\smallskip 

\tcc{Find one choice of $\lr_x^y$.}
\underline{QueryOneRMQ($x,y$)}:

\If{$L_y\neq -1$ and $R_x \neq -1$ and $L_y\leq R_x$} 
 {\Return{$\llrc\Bigl[RMQ\bigl(\llrc[L_y.. R_x]\bigr)\Bigr]$\;}}
\lElse {\Return{$(-1,-1)$\label{line:one}\tcp*{$\lr_x^y$ does not exist.}}}

\smallskip 

\tcc{Find all choices of $\lr_x^y$.}
\underline{QueryAllRMQ($x,y$)}

\If{$L_y\neq -1$ and $R_x \neq -1$ and $L_y \leq R_x$} 
 {$m \leftarrow RMQ\bigl(\llrc[L_y .. R_x]\bigr)$\; 
   $weight \leftarrow \llrc[m].end-\llrc[m].start+1$\tcp*{$|\lr_x^y|$}
   FindAllRMQ($L_y,R_x,weight$)\tcp*{Recursive searches start}}
\lElse {\Return{$(-1,-1)$\tcp*{$\lr_x^y$ does not exist.}}}

\smallskip

\underline{FindAllRMQ($\ell,r,weight$)} \tcp*[f]{Helper function}

 $m\leftarrow$ $RMQ\bigr(\llrc[\ell .. r]\bigr)$\;

\If{$\llrc[m].end-\llrc[m].start + 1 < weight$}{\Return\tcp*{Recursion exits.}}

Print $\llrc[m]$ \tcp*{One choice of $\lr_x^y$ is found.}

\If{$\ell\leq m-1$}
{FindAllRMQ($\ell,m-1,weight$) \tcp*{New recursive search.}}

\If{$r\geq m+1$}
{FindAllRMQ($m+1,r,weight$) \tcp*{New recursive search.}}
}
\end{algorithm}

\begin{theorem}
  We can construct an indexing structure for a string $S$ of size $n$
  using $O(n)$ time and  space, such that any future generalized $\lr$
  query can be answered in $O(1)$ time. If There exist multiple
  choices for the $\lr$ being queried for, ties are resolved
  arbitrarily.
\end{theorem}

\begin{proof}
  (1) The suffix array of $S$ can be constructed by existing
  algorithms using $O(n)$ time and space (e.g.,
  \cite{KA-SA2005}). After the suffix array is constructed, the rank
  array can be trivially created using $O(n)$ time and space.  We can
  then use the suffix array and the rank array to construct the lcp
  array using another $O(n)$ time and space~\cite{KLAAP01}.  (2) Given
  the rank array and the lcp array, we can construct the $\llrc$ array
  of useful LLRs using $O(n)$ time and space
  (Lemma~\ref{lem:comp-useful} and Algorithm~\ref{algo:useful}). (3)
  Given the LLRc array, we can compute the $L$ and $R$ arrays using
  another $O(n)$ time and space (Lemma~\ref{lem:range} and
  Algorithm~\ref{algo:range}). (4) We then create the RMQ structure
  for the LLRc array using another $O(n)$ time and space, using
  existing techniques (e.g.,~\cite{FH-2006CPM,HT-84SICOMP}).  So, the
  total time and space cost for building the indexing structure is
  $O(n)$.  By using this RMQ indexing structure and the pre-computed
  $L$ and $R$ arrays, we can answer any future generalized LR query in
  $O(1)$ time (The {\tt QueryOneRMQ} function in
  Algorithm~\ref{algo:rmq}). If there exist multiple choices for the
  LR being searched for, ties are resolved arbitrarily, as is
  determined by the RMQ structure.  
\end{proof}

\subsection{Find all choices of any LR.}
Upon receiving a query for $\lr_x^y$, we first perform an RMQ over
range $\llrc[L_y .. R_x]$ if such range exists; otherwise, it
means $\lr_x^y$ does not exist, and we stop. 
Suppose the range $\llrc[L_y .. R_x]$ is valid and its RMQ returns
$m$, the array index of the heaviest element in the range,
then
$\llrc[m]$ is one of the choices for $\lr_x^y$ and $|\lr_x^y| =
\llrc[m].end - \llrc[m].start+1$. If $\lr_x^y$ has other
choices, those choices must be existing in the union of the ranges 
$\llrc[L_y.. m-1]$ and $\llrc[m+1.. R_x]$. We can find those
choices of $\lr_x^y$ by recursively performing an RMQ on each of those two
ranges. The recursion will exit, if the element returned by RMQ has
a weight smaller than $|\lr_x^y|$ or the range for RMQ is invalid. 
The {\tt QueryAllRMQ} function in Algorithm~\ref{algo:rmq}
shows the pseudocode of this procedure for finding all choices of an LR query.

\begin{example}[Figure~\ref{fig:geo}]
Given the LLRc array of
the example string in Figure~\ref{fig:geo}, Algorithm~\ref{algo:range} computes the $L$ and $R$
arrays.

\remove{
{\scriptsize
  \begin{tabular}{c||c|c|c|c|c|c|c|c|c}
  $i$ & $1$ & $2$ & $3$ & $4$ & $5$ & $6$ & $7$ & 
  $8$ & $9$ \\
\hline
\hline
\tallstrut$L_i$ & $1$&$1$&$1$&$1$&$1$&$2$&$2$&$2$&$3$\\
\hline
\tallstrut$R_i$ &$1$&$1$&$1$&$1$&$2$&$2$&$3$&$3$&$3$
\end{tabular}

  \begin{tabular}{c||c|c|c|c|c|c|c|c}
  $i$ & $10$ & $11$ &$12$ &$13$ &$14$ &$15$ 
  &$16$ &$17$ \\
\hline
\hline
\tallstrut$L_i$ &$3$&$3$&$3$&$3$&$4$&$5$&$5$&$5$\\
\hline
\tallstrut$R_i$ &$4$&$5$&$5$&$5$&$5$&$5$&$5$&$5$
\end{tabular}
}
}

{\footnotesize
  \begin{tabular}{c||@{}c@{}|@{}c@{}|@{}c@{}|@{}c@{}|@{}c@{}|@{}c@{}|@{}c@{}|@{}c@{}|@{}c@{}|@{}c@{}|@{}c@{}|@{}c@{}|@{}c@{}|@{}c@{}|@{}c@{}|@{}c@{}|@{}c@{}}
  i \ & \ $1$\ \ & \ $2$\ \ & \ $3$\ \ & \ $4$\ \  & \ $5$\ \  & \ $6$\ \  & \ $7$\ \  & 
  \ $8$\ \  & \ $9$\ \  & \ $10$\ \ & \ $11$\ \ & \ $12$\ \ & \ $13$\ \ & \ $14$\ \
  & \ $15$\ \  
  & \ $16$\ \  & \ $17$\ \   \\
\hline
\hline
$L_i$ \ & $1$&$1$&$1$&$1$&$1$&$2$&$2$&$2$&$3$ & $3$&$3$&$3$&$3$&$4$&$5$&$5$&$5$\\
\hline
$R_i$ \ &$1$&$1$&$1$&$1$&$2$&$2$&$3$&$3$&$3$ &$4$&$5$&$5$&$5$&$5$&$5$&$5$&$5$
\end{tabular}
}

Upon receiving the query $\lr_{11}^{12}$, we first use the $L$ and $R$
arrays to retrieve the range $[L_{12},R_{11}] = [3,5]$, which is a
valid range for RMQ. Then we perform 
$RMQ\bigl(\llrc[3..5]\bigr)$ of Search \circled{A}
and either $3$ or $5$ can be returned, because both $\llrc[3]$ and
$\llrc[5]$ are the heaviest elements in the range $\llrc[3..
5]$. Suppose $3$ is returned and is saved in $m$, then we get
$\llrc[3]$ as one choice for $\lr_{11}^{12}$ and $|\lr_{11}^{12}| =
|\llrc[3]| = 7$.

Then, we will find other choices for $\lr_{11}^{12}$ by performing a
recursive search on each of the ranges $[L_{12},m-1] = [3,2]$ and
$[m+1,R_{11}] = [4,5]$.  The first range is invalid, so the search
exits (meaning Search $\circled{C}$ in Figure~\ref{fig:geo} will not
be performed).  The search on the second 
range $[4,5]$ (corresponding to Search $\circled{B}$ in
Figure~\ref{fig:geo}), which is valid,
will launch $RMQ\bigl(\llrc[4,5]\bigr)$. The RMQ
will return $5$. Since $|\llrc[5]| =
|\lr_x^y| = 7$, $\llrc[5]$ is another choice for $\lr_x^y$.

Then, the search on the range $[4,5]$ will launch its own two recursive
searches on the ranges $[4,5-1]=[4,4]$ and $[5+1,5] = [6,5]$. The
search on the first range will find the heaviest element's weight is
less than $|\lr_x^y|$, so the search stops. Because the second range
is invalid, the recursive search on that range will stop immediately.

At this point, all choices for $\lr_{11}^{12}$, which are $\llrc[3]$
and $\llrc[5]$, have been found.
\end{example}

Clearly, the same indexing structure is used by all RMQ's in the
search for all choices of $\lr_x^y$. Further, 
there are no more than $2 \cdot occ+1$ RMQ's in the
finding of all choices of one LR, where $occ$ is the number of choices
of the LR.  Because each RMQ takes $O(1)$ time, we get the following
theorem.

\begin{theorem}
  We can construct an indexing structure for a string $S$ of size $n$
  using $O(n)$ time and space, such that by using the indexing
  structure, we can find all choices of any generalized LR in $O(occ)$ time, where
  $occ$ is the number of choices of the LR being queried for.
\end{theorem}


\begin{figure*}[t]
\begin{tabular}{cc}
\includegraphics[scale=0.995]{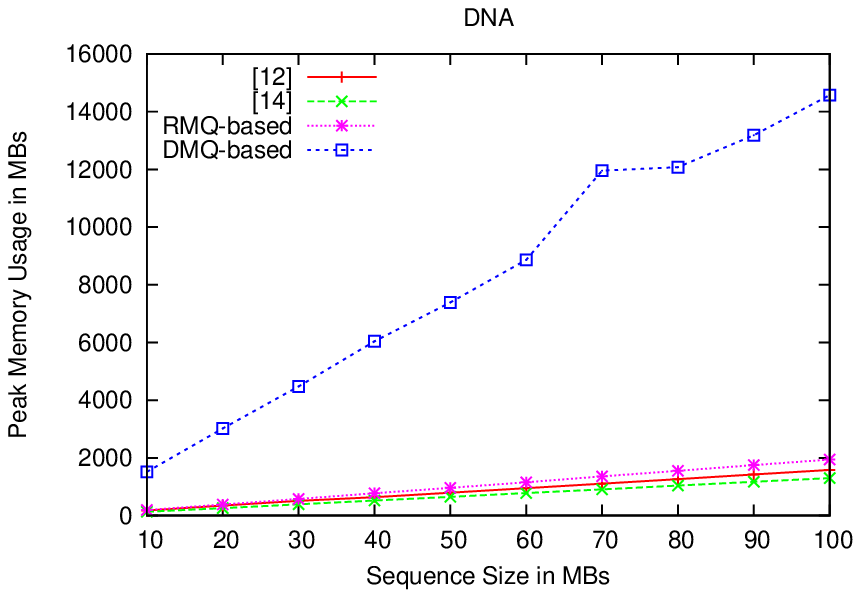} &
\includegraphics[scale=0.995]{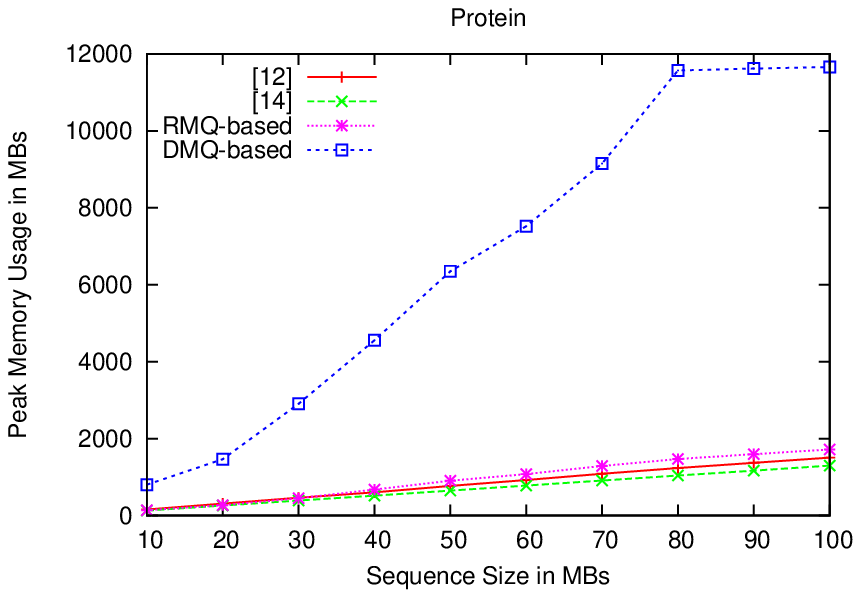} 
\end{tabular}
\caption{Peak memory usage of different proposals for
DNA and Protein strings of different sizes}
\label{fig:space}
\end{figure*}


\begin{figure*}[t]
\begin{tabular}{cc}
\includegraphics[scale=0.995]{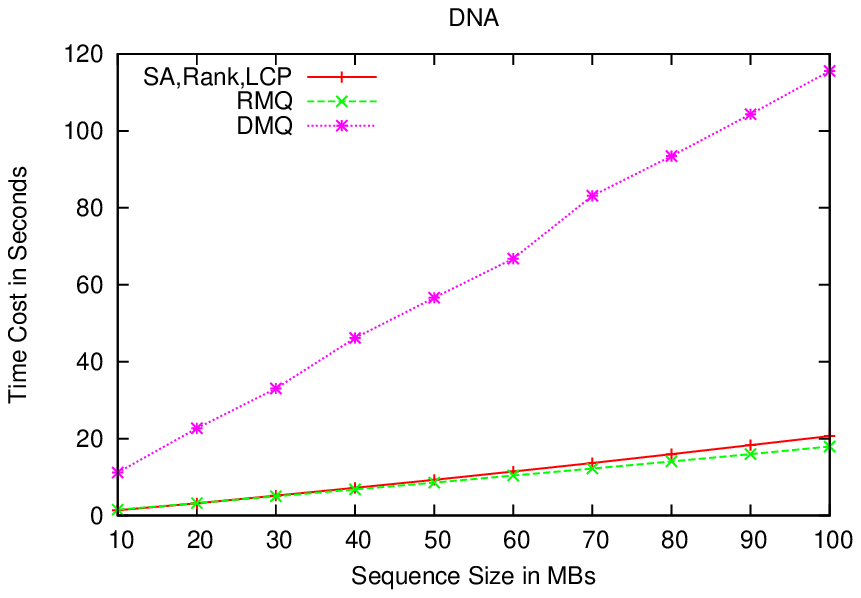} &
\includegraphics[scale=0.995]{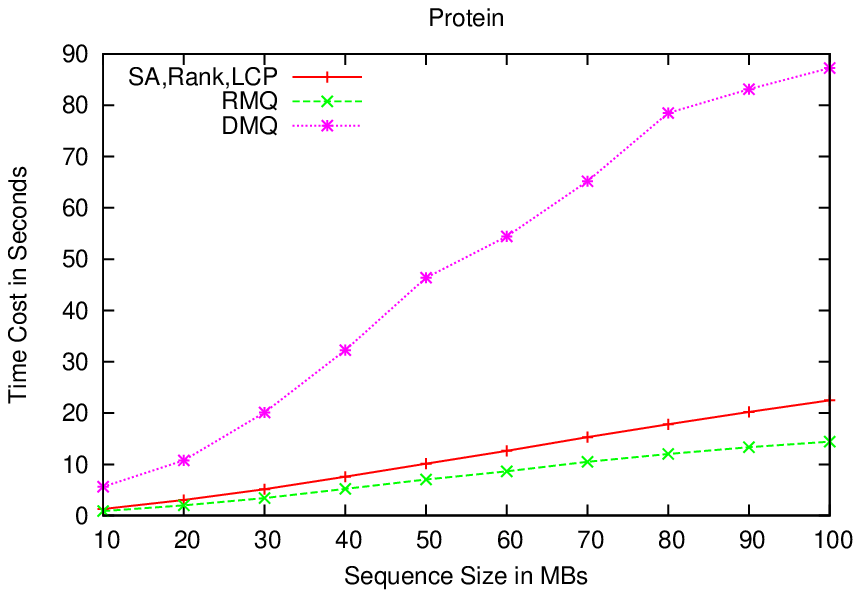} 
\end{tabular}
\caption{Indexing structure construction time for DNA and
  Protein strings of different sizes}
  \label{fig:time-index}
\end{figure*}


\begin{figure*}[t]
\begin{tabular}{cc}
\includegraphics[scale=0.995]{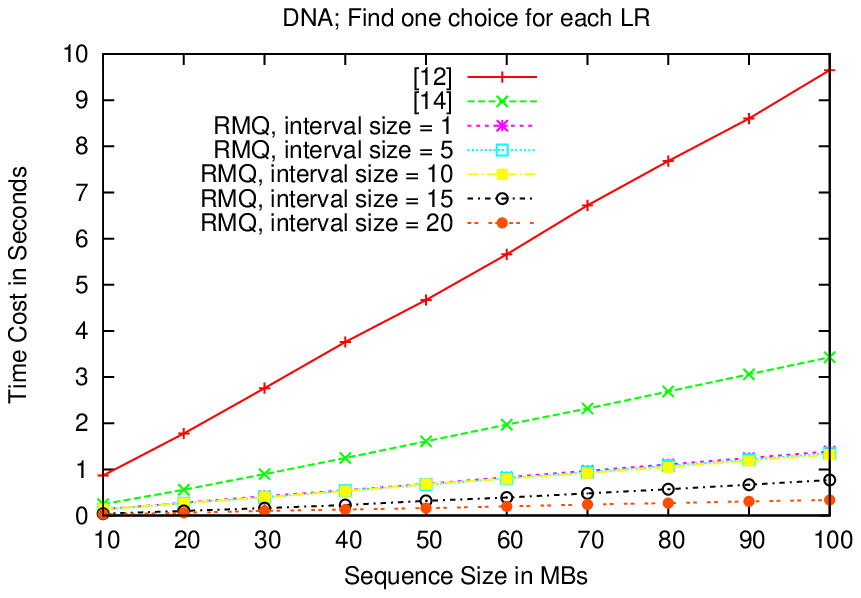} & 
\includegraphics[scale=0.995]{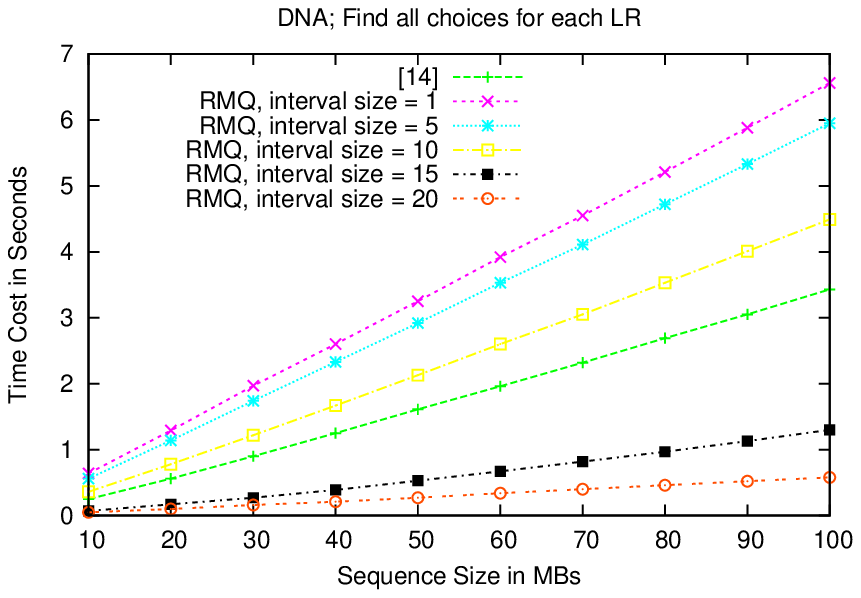} \\
\includegraphics[scale=0.995]{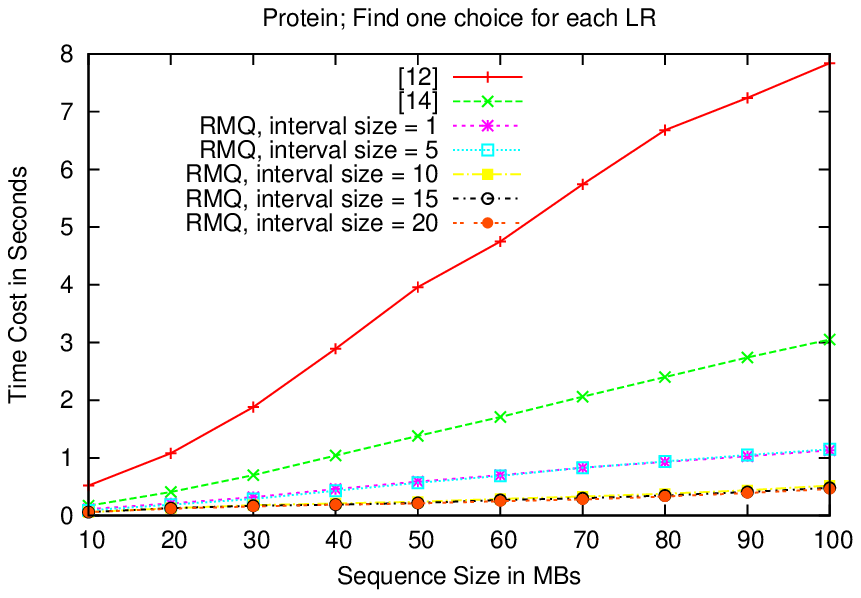} & 
\includegraphics[scale=0.995]{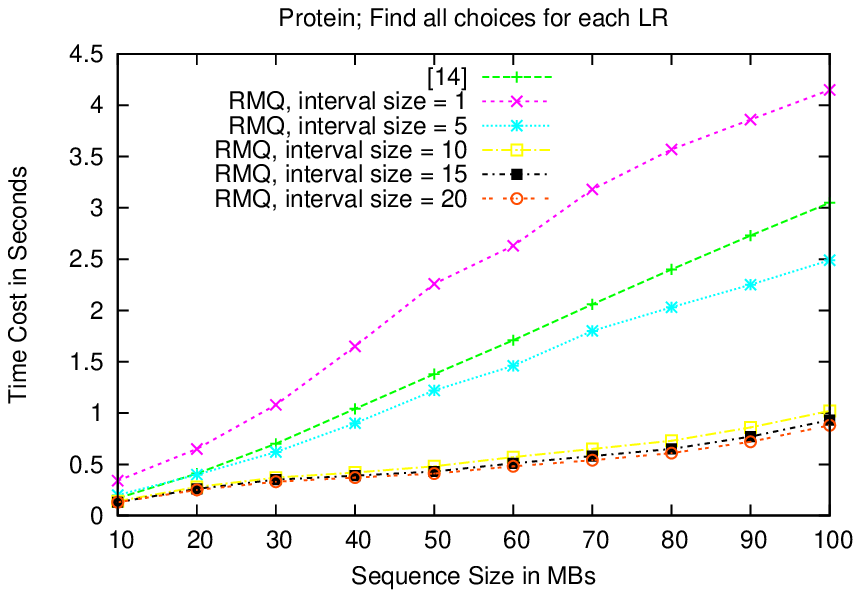} 
\end{tabular}
\caption{Query time of different proposals for DNA and Proteins
  strings of different sizes}
  \label{fig:time-query}
\end{figure*}


\section{Implementation and Experiments}
\label{sec:exp}
We implement our proposals in {\tt C++}, using the library binary of the
implementation of the DMQ and RMQ structures
from~\cite{HPT-spire2014}. Our implementation is generic in that it
does not assume the alphabet size of the underlying string, and thus
supports LR queries over different types of strings.

We compare the performance of our
proposals with the prior works including the optimal $O(n)$ time and
space solution from~\cite{IKX-repeat-CORR2015}
and the suboptimal sequential algorithm presented
in~\cite{TX-GPU-DASFAA2015}.
Note that all prior works can only answer point queries.
All programs involved in the experiments use the same
\texttt{libdivsufsort}\footnote{{\tt https://code.google.com/p/libdivsufsort}.}
library for the suffix array construction, and are compiled by {\tt
  gcc 4.7.2} with {\tt -O3} option.

We conduct our experiments on a GNU/Linux machine with kernel version
3.2.51-1.  The computer is equipped with an Intel Xeon 2.40GHz E5-2609
CPU with 10MB Smart Cache and has 16GB RAM.  All experiments are
conducted on real-world datasets including the {\tt DNA} and {\tt
  Protein} strings, downloaded from the Pizza\&Chili
Corpus\footnote{{\tt http://pizzachili.dcc.uchile.cl/texts.html}}.
The datasets we use are the two $100$MB {\tt DNA} and {\tt Protein}
pure ASCII text files, each of which thus represents a string of
$100\times 1024 \times 1024 = 104,857,600$ characters. Any other
shorter strings involved in our experiments are prefixes of certain
lengths, cut from the $100$MB strings. 

\remove{

\smallskip 
\noindent
{\bf Measurement.}
In order to focus on the comparison of the algorithmic performance of
difference proposals, we do not save the output in all experiments. We
use the command \texttt{/usr/bin/time -f} ``\texttt{\%M}'' provided by
Linux to report the peak memory usage of a run, and use the system
clock time to monitor the time cost.

}

\remove{
\begin{figure}[t]
\begin{center}
\includegraphics[scale=0.995]{stat_lr_all.eps}
\end{center}
\caption{The total number of LRs that cover each of the $n-\delta+1$
  intervals in 100MB DNA and Protein strings, where $n=100\times
  1024 \times 1024$ is the string size and $\delta=1,5,\ldots,30$ is
  the interval size.}
  \label{fig:stat}
\end{figure}
}

\remove{

\begin{figure}[t]
\begin{center}
\includegraphics[scale=0.995]{time_query_rmq_interval_DNA_protein_all.eps}
\end{center}
\caption{The total number of LRs that cover each of the $n-\delta+1$
  intervals in 100MB DNA and Protein strings, where $n=100\times
  1024 \times 1024$ is the string size and $\delta=1,5,\ldots,30$ is
  the interval size.}
  \label{fig:rmq-query}
\end{figure}

}


\subsection{Space}
\label{subsec:space}

Here, we measure the peak memory usage of different
proposals, using the Linux command \texttt{/usr/bin/time -f}
``\texttt{\%M}'' that captures the maximum resident set size of a
process during its lifetime.  We do not save the output in the RAM in
order to focus on the comparison of the memory usage of the
algorithmics.  It is also because practitioners often flush the
outputs directly to disk files for future reuse.

Figure~\ref{fig:space} shows the peak memory usage of different
proposals that process DNA and protein strings of different sizes.
It is worth noting that, by design, the memory usage of each proposal
is independent from the query type, such as finding one choice vs.\
all choices of an LR, point query vs.\ interval query. We have the
following main observations:

\noindent
-- All proposals show the linearity of their space usage 
over  string size.

\noindent
-- Our DMQ-based proposal uses much more memory space than other
  proposals. It is mainly caused by the high space demand from the DMQ
  structure.

\noindent
-- Our RMQ-based proposal uses nearly the same amount of memory space  as that
  of prior works, while significantly improving the usability of
  the technique by providing the functionality of interval queries.  

\subsection{Time}
\label{subsec:time}

Figure~\ref{fig:time-index}
shows the construction time  of the indexing structures used by different proposals. 
Note that all proposals need to construct the suffix array, rank
array, and the lcp array of the given string, and our proposals
further use these auxiliary arrays to construct the DMQ and RMQ
structures for interval queries. 
The following are the main observations: 

\noindent
-- The construction of the DMQ structure takes much more time than 
that of  the auxiliary arrays and the RMQ structure. 

\noindent
-- Both the auxiliary array and RMQ structure clearly show
the  linearity in the their construction time over string size.

\noindent
-- The construction of the RMQ structure takes less time than the
  construction of the auxiliary arrays, making our RMQ-based proposal
  practical while supporting 
  interval queries.

Figure~\ref{fig:time-query} shows the time cost of various types of query.
Our DMQ-based proposal is so slow in query response
that we do not include it in the figure.  For point queries, we plot
the total time cost for all the point queries over all $n$ string
positions, where $n$ is the string size. For interval queries with
interval size $\delta$, we plot the total time cost for all the
interval queries over all $n-\delta+1$ intervals of the string. Note
that only point queries are involved in the experiments with the
proposals from~\cite{IKX-repeat-CORR2015}
and~\cite{TX-GPU-DASFAA2015}, because they do not support interval
queries.  The two figures on the left show the case for finding only
one choice for each LR, whereas the two on the right show the case for
finding all choices for each LR.  Because the proposal
from~\cite{IKX-repeat-CORR2015} does not support the finding of all
choices, it is not included in the two figures on the right side.
The following are the main observations:

\noindent
-- All proposals show the clear linearity of the total query time
  cost, meaning the amortized $O(1)$ time cost for each query.

\noindent
-- In the setting of finding one choice for each LR (the two
  figures on the left of Figure~\ref{fig:time-query}), our RMQ-based
  proposal is the fastest regarding the per-query response time,
  including both point query and interval query! Further, our RMQ-based
  proposal's interval query response becomes even faster,
  when interval size increases. That is because a longer
  interval is covered by fewer number of repeats, reducing the
  search space size for finding the LR covering the interval.

\noindent
-- In the setting of finding all choices for each LR (the two
  figures on the right side of Figure~\ref{fig:time-query}): 

\begin{itemize}
\item 
  For point query, our RMQ-based proposal is a little slower
  than~\cite{TX-GPU-DASFAA2015} due to the following reason.  On
  average, an LR point query returns more choices than an interval
  query.  Our technique needs to make a query to the index for finding
  every single choice, whereas the technique
  in~\cite{TX-GPU-DASFAA2015} only needs one extra ``walk'' for
  finding all choices for a particular LR point query.  Even though
  our technique is faster than~\cite{TX-GPU-DASFAA2015} for finding
  one choice (the two figures on the left), when a particular point
  query has many choice, our technique can become slower in finding all choices.

\item
  As interval size increases, our RMQ-based proposal becomes faster,
  because a longer interval on average has fewer choices for its LR,
  making our technique have fewer queries to its index. Our
  technique's interval query can be even faster than the point query
  by~\cite{TX-GPU-DASFAA2015} in finding all choices when interval
  size increases. For example, it is true,  when interval size becomes $\geq 15$ for DNA
  string (top-right figure) and $\geq 5$ for protein string
  (bottom-right figure).

\remove{
3. DNA string has more LR until interval size becomes large than 10,
so the point query and interval queries with interval size 5 and 10
are slower than the point query of [19]. For protein string, the \#LR
is smaller than those of the DNA strings, so only point query is
slower. Supported by [19]. Interval quires of the same interval size
for DNA data is slower than those of protein data, because more
choices for LR when interval size is less than 15. supported by figure 5. 
}

\end{itemize}

\section{Conclusion}
\label{sec:conclusion}
We generalized the longest repeat query on a string from point query
to interval query and proposed both time and space optimal solution
for interval queries. Our approach is different from prior work which
can only handle point queries. Using the insight
from~\cite{HPT-spire2014}, we proposed an indexing structure that
can be built on top of the string using time and space linear of the
string size, such that any future interval queries can be answered in
$O(1)$ time. 
We implemented our proposals without assuming the alphabet size of the string, 
making it useful for different types of strings. 
An interesting future work is to parallelize our proposal so as to take
advantage of the modern multi-core and multi-processor computing
platforms, such as the general-purpose graphics processing units.

\remove{
\section*{Acknowledgments}
The author would like to thank Professor Yufei Tao from Chinese
University of Hong Kong for sharing their manuscript on 
Shortest Unique Queries on Strings~\cite{HPT-spire2014},
which sparked the idea of the problem
and its solution presented in this paper.
}

\balance

\bibliographystyle{IEEEtran}
\bibliography{bibjsv,repeat,pm}

\begin{thebibliography}{10}
\providecommand{\url}[1]{#1}
\csname url@samestyle\endcsname
\providecommand{\newblock}{\relax}
\providecommand{\bibinfo}[2]{#2}
\providecommand{\BIBentrySTDinterwordspacing}{\spaceskip=0pt\relax}
\providecommand{\BIBentryALTinterwordstretchfactor}{4}
\providecommand{\BIBentryALTinterwordspacing}{\spaceskip=\fontdimen2\font plus
\BIBentryALTinterwordstretchfactor\fontdimen3\font minus
  \fontdimen4\font\relax}
\providecommand{\BIBforeignlanguage}[2]{{%
\expandafter\ifx\csname l@#1\endcsname\relax
\typeout{** WARNING: IEEEtran.bst: No hyphenation pattern has been}%
\typeout{** loaded for the language `#1'. Using the pattern for}%
\typeout{** the default language instead.}%
\else
\language=\csname l@#1\endcsname
\fi
#2}}
\providecommand{\BIBdecl}{\relax}
\BIBdecl

\bibitem{HPT-spire2014}
X.~Hu, J.~Pei, and Y.~Tao, ``Shortest unique queries on strings,'' in
  \emph{Proceedings of International Symposium on String Processing and
  Information Retrieval (SPIRE)}, 2014, pp. 161--172.

\bibitem{Gus97}
D.~Gusfield, \emph{Algorithms on strings, trees and sequences: computer science
  and computational biology}.\hskip 1em plus 0.5em minus 0.4em\relax Cambridge
  University Press, 1997.

\bibitem{McC1993}
E.~H. McConkey, \emph{Human Genetics: The Molecular Revolution}.\hskip 1em plus
  0.5em minus 0.4em\relax Boston, MA: Jones and Bartlett, 1993.

\bibitem{LW06}
X.~Liu and L.~Wang, ``Finding the region of pseudo-periodic tandem repeats in
  biological sequences,'' \emph{Algorithms for Molecular Biology}, vol.~1,
  no.~1, p.~2, 2006.

\bibitem{Mar83}
H.~M. Martinez, ``An efficient method for finding repeats in molecular
  sequences,'' \emph{Nucleic Acids Research}, vol.~11, no.~13, pp. 4629--4634,
  1983.

\bibitem{BDH09}
V.~Becher, A.~Deymonnaz, and P.~A. Heiber, ``Efficient computation of all
  perfect repeats in genomic sequences of up to half a gigabyte, with a case
  study on the human genome,'' \emph{Bioinformatics}, vol.~25, no.~14, pp.
  1746--1753, 2009.

\bibitem{KVX-tcbb2012}
M.~O. Kulekci, J.~S. Vitter, and B.~Xu, ``Efficient maximal repeat finding
  using the burrows-wheeler transform and wavelet tree,'' \emph{IEEE
  Transactions on Computational Biology and Bioinformatics (TCBB)}, vol.~9,
  no.~2, pp. 421--429, 2012.

\bibitem{BBO-spire2012}
T.~Beller, K.~Berger, and E.~Ohlebusch, ``Space-efficient computation of
  maximal and supermaximal repeats in genome sequences,'' in \emph{Proceedings
  of the 19th International Conference on String Processing and Information
  Retrieval (SPIRE)}, 2012, pp. 99--110.

\bibitem{BIMSTTT2007}
A.~Bakalis, C.~S. Iliopoulos, C.~Makris, S.~Sioutas, E.~Theodoridis, A.~K.
  Tsakalidis, and K.~Tsichlas, ``Locating maximal multirepeats in multiple
  strings under various constraints,'' \emph{Computer Journal}, vol.~50, no.~2,
  pp. 178--185, 2007.

\bibitem{ISY2009}
C.~S. Iliopoulos, W.~F. Smyth, and M.~Yusufu, ``Faster algorithms for computing
  maximal multirepeats in multiple sequences,'' \emph{Fundamenta Informaticae},
  vol.~97, no.~3, pp. 311--320, 2009.

\bibitem{IS2011}
L.~Ilie and W.~F. Smyth, ``Minimum unique substrings and maximum repeats,''
  \emph{Fundamenta Informaticae}, vol. 110, no. 1-4, pp. 183--195, 2011.

\bibitem{IKX-repeat-CORR2015}
A.~M. \.{I}leri, M.~O. K\"{u}lekci, and B.~Xu, ``On longest repeat queries,''
  \emph{\ http://arxiv.org/abs/1501.06259}.

\bibitem{SOG2012-IC}
T.~Schnattinger, E.~Ohlebusch, and S.~Gog, ``Bidirectional search in a string
  with wavelet trees and bidirectional matching statistics,'' \emph{Information
  and Computation}, vol. 213, pp. 13--22, Apr. 2012.

\bibitem{TX-GPU-DASFAA2015}
Y.~Tian and B.~Xu, ``On longest repeat queries using {GPU},'' in
  \emph{Proceedings of the 20th International Conference on Database Systems
  for Advanced Applications (DASFAA)}, 2015, pp. 316--333.

\bibitem{PWY-ICDE2013}
J.~Pei, W.~C.~H. Wu, and M.~Y. Yeh, ``On shortest unique substring queries,''
  in \emph{Proceedings of the 2013 IEEE International Conference on Data
  Engineering (ICDE)}, 2013, pp. 937--948.

\bibitem{TIBT2014}
K.~Tsuruta, S.~Inenaga, H.~Bannai, and M.~Takeda, ``Shortest unique substrings
  queries in optimal time,'' in \emph{Proceedings of International Conference
  on Current Trends in Theory and Practice of Computer Science (SOFSEM)}, 2014,
  pp. 503--513.

\bibitem{IKX-CPM2014}
A.~M. \.{I}leri, M.~O. K\"{u}lekci, and B.~Xu, ``Shortest unique substring
  query revisited,'' in \emph{Proceedings of the 25th Annual Symposium on
  Combinatorial Pattern Matching (CPM)}, 2014, pp. 172--181.

\bibitem{ikx-sus-tcs2015}
A.~M. \.{I}leri, M.~O. K{\"{u}}lekci, and B.~Xu, ``A simple yet time-optimal
  and linear-space algorithm for shortest unique substring queries,''
  \emph{Theoretical Computer Science}, vol. 562, no.~0, pp. 621 -- 633, 2015.

\bibitem{KA-SA2005}
P.~Ko and S.~Aluru, ``Space efficient linear time construction of suffix
  arrays,'' \emph{Journal of Discrete Algorithms}, vol.~3, no. 2-4, pp.
  143--156, 2005.

\bibitem{KLAAP01}
T.~Kasai, G.~Lee, H.~Arimura, S.~Arikawa, and K.~Park, ``Linear-time
  longest-common-prefix computation in suffix arrays and its applications,'' in
  \emph{Symposium on Combinatorial Pattern Matching}, 2001, pp. 181--192.

\bibitem{ST-2011pods}
C.~Sheng and Y.~Tao, ``New results on two-dimensional orthogonal range
  aggregation in external memory,'' in \emph{Proceedings of the Thirtieth ACM
  SIGMOD-SIGACT-SIGART Symposium on Principles of Database Systems (PODS)},
  2011, pp. 129--139.

\bibitem{FH-2006CPM}
J.~Fischer and V.~Heun, ``Theoretical and practical improvements on the
  rmq-problem, with applications to lca and lce,'' in \emph{Proceedings of the
  17th Annual Conference on Combinatorial Pattern Matching (CPM)}, 2006, pp.
  36--48.

\bibitem{HT-84SICOMP}
D.~Harel and R.~E. Tarjan, ``Fast algorithms for finding nearest common
  ancestors,'' \emph{SIAM Journal of Computing}, vol.~13, no.~2, pp. 338--355,
  May 1984.

\end{thebibliography}

\end{document}

\newpage

\appendix


\begin{table*}[h]
  \caption{Peak memory usage of different proposals for
DNA and Protein strings of different sizes}
\label{tab:space}
\medskip 
\begin{center}
  \begin{tabular}{r||r|r|r|r||r|r|r|r}
  String\ \, & \multicolumn{4}{c||}{Peak Memory Usage for DNA Strings}
  & \multicolumn{4}{c}{Peak Memory Usage for Protein Strings} \\
\cline{2-9}
\tallstrut Size (MB) & [8] & [19] & RMQ-based & DMQ-based & [8] & [19] & RMQ-based & DMQ-based\\
\hline 
\tallstrut
$10$ & $176$ & $131$ & $192$ & $1,515$ & $159$ & $131$ & $141$ & $803$\\
$20$ & $350$ & $261$ & $383$ & $3,020$ & $312$ & $261$ & $270$ & $1,468$\\
$30$ & $509$ & $391$ & $575$ & $4,477$ & $466$ & $391$ & $451$ & $2,907$\\
$40$ & $635$ & $521$ & $775$ & $6,039$ & $605$ & $521$ & $677$ & $4,558$\\
$50$ & $793$ & $651$ & $964$ & $7,392$ & $774$ & $651$ & $908$ & $6,350$\\
$60$ & $951$ & $781$ & $1,156$ & $8,865$ & $925$ & $781$ & $1,080$ & $7,523$\\
$70$ & $1,107$ & $911$ & $1,359$ & $11,957$ & $1,089$ & $911$ & $1,289$ & $9,150$\\
$80$ & $1,265$ & $1,041$ & $1,553$ & $12,076$ & $1,238$ & $1,041$ & $1,473$ & $11,570$\\
$90$ & $1,424$ & $1,171$ & $1,749$ & $13,182$ & $1,374$ & $1,171$ & $1,599$ & $11,623$\\
$100$ & $1,582$ & $1,301$ & $1,945$ & $14,573$ & $1,507$ & $1,301$ & $1,718$ & $11,661$\\ 
\hline
  \end{tabular}
\end{center}
\end{table*}

\begin{table*}[h]
  \centering
\caption{Indexing structure construction time (in seconds) for DNA and
  Protein strings of different sizes}
  \label{tab:time-index}
\medskip 
  \begin{tabular}{r||r|r|r||r|r|r}
String\,\, & \multicolumn{3}{c||}{Index construction time for DNA strings} & \multicolumn{3}{c}{Index construction time for Protein strings} \\
\cline{2-7}
\tallstrut size (MB) & SA, Rank, LCP & RMQ index & DMQ index & SA, Rank, LCP & RMQ index & DMQ index\\
\hline
\tallstrut
$10$ & $1.40$ & $1.53$ & $11.19$ & $1.28$ & $0.89$ & $5.61$\\
$20$ & $3.21$ & $3.26$ & $22.68$ & $3.06$ & $1.99$ & $10.75$\\
$30$ & $5.17$ & $5.03$ & $33.02$ & $5.15$ & $3.41$ & $20.10$\\
$40$ & $7.21$ & $6.80$ & $46.17$ & $7.59$ & $5.22$ & $32.27$\\
$50$ & $9.29$ & $8.58$ & $56.65$ & $10.12$ & $7.04$ & $46.36$\\
$60$ & $11.46$ & $10.41$ & $66.81$ & $12.63$ & $8.64$ & $54.42$\\
$70$ & $13.67$ & $12.22$ & $83.12$ & $15.30$ & $10.48$ & $65.20$\\
$80$ & $15.97$ & $14.09$ & $93.48$ & $17.80$ & $12.02$ & $78.48$\\
$90$ & $18.32$ & $15.99$ & $104.32$ & $20.22$ & $13.34$ & $83.13$\\
$100$ & $20.66$ & $17.95$ & $115.56$ & $22.48$ & $14.43$ & $87.25$\\
\hline
  \end{tabular}
\end{table*}

\end{document}